\documentclass[12pt]{article}
\usepackage{amsmath,amssymb,amsfonts,amsthm}
\newtheorem{theorem}{Theorem}[section]
\newtheorem{lemma}[theorem]{Lemma}

\newcommand{\Rt}{\mathbb{R}^3}
\newcommand{\Rtc}{\mathbb R^ 3\setminus\{0\}}

\begin{document}
\title{Conformally flat black hole initial data, with one cylindrical end}
\author{Mar\'\i a E. Gabach Clement$^1$\\
\\
$^1$Facultad de Matem\'atica, Astronom\'{i}a y F\'{i}sica, \\
Universidad Nacional de C\'ordoba,\\
 Ciudad Universitaria, (5000) C\'ordoba, Argentina.  }
\maketitle

\abstract{We give a complete analytical proof of existence and uniqueness of extreme-like black hole initial data for Einstein equations, which possess a cilindrical end, analogous to extreme Kerr, extreme Reissner N\"ordstrom, and extreme Bowen-York's initial data. This extends and refines a previous result \cite{dain-gabach09} to a general case of conformally flat, maximal initial data with angular momentum, linear momentum and matter.}

\section{Introduction}

Extreme black holes (i.e, black holes with maximum amount of angular momentum and electric charge per unit mass) have received increased attention during the last few years due to a number of reasons. Extreme solutions are good candidates to be used in getting insights concerning Penrose cosmic censorship hypothesis \cite{beigmurchadha}, because they lie on the frontier between black holes and naked singularities. For example, extreme Kerr black hole appears as a unique global minimum of the total mass, a property used to prove the inequality between mass and angular momentum \cite{dainang} which provides evidences in favor of cosmic
censorship. In addition, it has been suggested, both by observations \cite{komossa} and numerical calculations \cite{dainlousto} that black holes with large amounts of spin, play a dominant role in determining the high recoil velocities of black hole binaries. This important results came at a time when there was a diversity of positions regarding the very existence of nearly extreme black holes (see \cite{lovelace} and references therein for discussions on this subject). Some models \cite{volonteri}, stating that under the combined effects of accretion and binary coalescences, the spin distribution is heavily skewed towards fast rotating holes, confronted with  other models of black hole accretion not leading to large spins \cite{berti}, \cite{gammie}.

Extreme black hole initial data are also interesting from the geometric point  of view, since they present an asymptotically cylindrical end. Data with one cylindrical end have been much used in numerical evolutions, where they are known as \textit{trumpet} data (see \cite{hhobm}, \cite{lovelace}, \cite{B09}, \cite{HHM09} and references therein). It has been noticed that when compactified wormholes are evolved using the standard moving puncture method, the slices lose contact with the extra asymptotically flat wormhole end, and quickly asymptote to cylinders of finite areal radius. It was also seen that maximally sliced data with this cylindrical topology are indeed time independent in a moving puncture simulation.
 
In this work we continue the analytical study of extreme black hole initial data  started in \cite{dain-gabach09}, where the existence of extreme initial data built from the Bowen-York family of spinning black hole data was proven. It was shown that the \textit{non-extreme} solutions constitute a monoparametric family (for fixed angular momentum) of initial data, which were called $u_\mu$ for each positive value of the parameter $\mu$. Then the $\mu\to0$ limit solution was identified as the extreme one, due to its similarities with extreme Kerr and Reissner N\"ordstrom analogues. 

In this respect, what we want to show here is that the same phenomena occur in a wider class of black holes initial data, more especifically, in conformally flat, maximal initial data, without currents on the initial hypersurface, indicating that an extreme solution exists for each family of black hole initial data under the above hypotheses, and suggesting that the cylindrical character of one of the asymptotic ends is a general property among extreme black holes.

We remark that not only we deal with a richer class of initial data, but we also prove that the extreme solution is unique in the appropriate functional space, a key point lacking in \cite{dain-gabach09} for the simpler case of a spinning Bowen-York black hole. In addition, the proof of uniqueness will refine the singular behavior of the extreme solution near the origin.

The paper is organized as follows. First, in section 2 we present the problem and equations involved, and state our main result. We also discuss the more relevant geometric and physical features of the extreme initial data obtained. In section 3 we prove existence of this initial data through a limiting procedure and finally, in section 4 we prove that it is unique in the appropriate functional space.

\section{Settings and Main Result}
An initial data set for Einstein equations \cite{Bar-ise04} consists of a Riemannian metric $\bar{g}$ and a symmetric 2-tensor field $\bar K$ on a three dimensional manifold $M$. These data must satisfy Einstein constraints, linking the metric $\bar{g}$ on $M$ with the extrinsic curvature $\bar K$ of $M$ when seen as a submanifold imbedded in the spacetime. As equations on $M$, these constraints read
\begin{equation}\label{ham}
R(\bar{g})-\bar K\cdot \bar K+(\mbox{tr} \bar K)^2=16\pi\bar{\rho}
\end{equation}
\begin{equation}\label{mom}
\bar{\nabla}\cdot \bar K-\bar{\nabla}\mbox{tr} \bar K=8\pi\bar{j}
\end{equation}
where all derivatives and dot products are computed with respect to $\bar g$, $R(\bar g)$ is the Ricci scalar associated to $\bar g$, $\bar \rho$ is the energy density and $\bar j$, the currents on $M$. 

The method we will use to treat these constraints is the conformal method, which allows us to turn the hamiltonian constraint (\ref{ham}) into an elliptic equation for a scalar function $\Phi$ by considering the metric $\bar{g}$ as given up to a conformal factor.
\par
We will only consider maximal surfaces $M$ (that is, hypersurfaces such that the extrinsic curvature has vanishing trace) with no currents, and which are conformally flat. Because of the conformal invariance of the momentum constraint (\ref{mom}), we end up with the task of specifying a traceless and divergence-less tensor field $K$ on $M$ and the energy density of sources on the slice, $\rho$. Then, we just need to solve the constraints for $\Phi$. The (physical) initial data will be given by
\begin{equation}
\bar{g}_{ij}=\Phi^4 \delta_{ij}
\end{equation}
\begin{equation}
\bar K_{ij}=\Phi^{-2}K_{ij}
\end{equation}
where $\delta$ is the 3-euclidean metric, and $\Phi$ satisfies Lichnerowicz equation
\begin{equation}\label{lich}
\Delta \Phi=-\frac{K^2}{8\Phi^7}-\frac{2\pi\rho}{\Phi^3}:=F(x,\Phi)\;\;\;\;\mbox{in}\;M.
\end{equation}
In this equation, all derivatives and dot products are referred to the flat metric in $\mathbb R ^3$ and the scaled sources are 
\begin{equation}\label{scaling}
\rho=\bar\rho\Phi^{8}.
\end{equation}
This scaling includes, for instance, generic fluid sources with no independent field equations, electromagnetic sources and Yang-Mills fields \cite{chisy}.

We also set $M=\Rtc$, with spherical coordinates $r,\,\theta,\,\phi$, impose the energy condition $\rho\geq0$ and asymptotic flatness, both at infinity and at the origin. This last requirement is accomplished by defining a new function $u_\mu$ in all $\mathbb R ^3$, introducing a new positive parameter $\mu$ through the expression
 \begin{equation}\label{fimu}
\Phi:=\Phi_\mu= 1+\frac{\mu}{2r}+u_\mu\quad\quad\mu>0
 \end{equation}
and demanding $u_\mu$ to go to zero at infinity and to be well defined at the origin.

The corresponding equation for $u_\mu$ is
\begin{equation}\label{ecgeneralumu}
\Delta u_{\mu}=-\frac{K^2}{8\left(1+\frac{\mu}{2r}+u_{\mu}\right)^{7}}-\frac{2\pi\rho}{\left(1+\frac{\mu}{2r}+u_{\mu}\right)^3}=F(x,\Phi_\mu)\;\;\;\;\mbox{in}\;R^3.
\end{equation}

In \cite{df} it has been proven that any smooth, traceless solution of $\nabla\cdot K=0$ in $\Rt\setminus\{0\}$ is of the form  
\begin{equation}\label{kgenerallq}
 K=K_P+K_J+K_A+K_G+K_\lambda
\end{equation}
where  the first four terms on the right hand side are given by
\begin{equation}
 K^{ab}_G=\frac{3}{2r^4}\left(-G^an^b-G^bn^a-(\delta^{ab}-5n^an^b)G^cn_c\right)
\end{equation}
\begin{equation}\label{kj}
 K^{ab}_J=\frac{3}{r^3}\left(n^a\epsilon^{bcd}J_cn_d+n^b\epsilon^{acd}J_cn_d\right)
\end{equation}
\begin{equation}\label{ka}
 K^{ab}_A=\frac{A}{r^3}\left(3n^an^b-\delta^{ab}\right)
\end{equation}
\begin{equation}\label{kp}
K^{ab}_P=\frac{3}{2r^2}\left(P^an^b+P^bn^a-(\delta^{ab}-n^an^b)P^cn_c\right).
\end{equation}
In these expressions, $A\geq0$ is a constant, $J^a$ is the angular momentum of the data, $P^a$ is its linear momentum at infinity, and $G^a$, the linear momentum at the origin. The last term, $K_\lambda$ is a symmetric 2-tensor depending on a scalar function $\lambda$ and can be made smooth or smooth with compact support by suitable choices of $\lambda$ in the form $\lambda=\lambda_1+\frac{\lambda_2}{r}$ with $\lambda_1$ and $\lambda_2$ $C^\infty$ functions in $\Rt$ (see theorem 4.3 of \cite{df}). Since $K_\lambda$ can be completely controlled through appropriate forms for $\lambda$, in what follows, we will omit this term in all calculations. Moreover, we will restrict attention to the case in which $G\equiv0$, that is, 2-tensors having the form
\begin{equation}\label{kgeneral}
 K=K_P+K_J+K_A.
\end{equation}
Note  that in the previous work \cite{dain-gabach09}, it has only been considered the case $K=K_J$ (spinning Bowen-York black hole's initial data), and no matter.

For later use, it is convenient to write here the explicit form of $K^2$ 
\begin{eqnarray}\label{K2}
K^2=6\frac{A^2+3J^2\sin^2\theta}{r^6}+\frac{12AP^a n_a+18\epsilon_{abc} n^a P^bJ^c}{r^5}+\nonumber\\+\frac{9(P^2+2(P^an_a)^2)}{2r^4}
\end{eqnarray}
where $ n^a:=x^a/r$ is the radial unit vector, $x^a$ are Cartesian coordinates on $\Rt$, and we have oriented our coordinate axes so that $J^a$ lies along the $z$ direction. 

In what follows, we also assume that the energy density $\rho$ has an appropriate fall off behavior at both ends (see \cite{muryork} where matter requirements are discussed in the context of the finiteness of the total mass) by defining a bounded, non-negative, regular function $\sigma$ on $\mathbb R ^3$ such that $\rho$ has the form
\begin{equation}\label{rho}
 \rho=\frac{\sigma}{8\pi r^4}.
\end{equation}
Under the hypothesis assumed on matter and $K$, equation (\ref{ecgeneralumu}) is regular in $\Rt$, when $\mu>0$, and there exists a unique positive $C^1(\Rt)$ solution $u_\mu$ for each $\mu>0$. Since we want to investigate the limit $\mu\rightarrow0$, we define the extreme solution $u_0$ as the solution to the singular equation
\begin{equation}\label{ecgeneralu0}
 \Delta u=-\frac{K^2}{8(1+u)^7}-\frac{2\pi\rho}{(1+u)^3}
\end{equation}
which will be constructed as the limit
\begin{equation}
u_0:=\lim_{\mu\to0}u_\mu
\end{equation}
in the sequence of solutions to the non-extreme equations \eqref{ecgeneralumu}. Note that equation \eqref{ecgeneralu0} is obtained from \eqref{ecgeneralumu} if we set $\mu=0$. We will find that this singular limit indeed exists, that it is the unique solution to \eqref{ecgeneralu0}, and that it has some properties similar to the more familiar cases of extreme Kerr and extreme Reissner N\"ordstrom solutions.
\par
As mentioned in \cite{dain-gabach09}, the natural functional spaces arising in this problem are weighted Sobolev spaces $H'^{2}_\delta$ \cite{Bartnik86}. In these spaces is where existence and uniqueness of solution is proven. They are defined as the completion of $C^\infty$ functions with compact support away from the origin under the norms (we focus on the case $p=2$ and dimension $3$)
\begin{equation}
\label{eq:35}
\Vert f\Vert_{L'^{2}_\delta}=\left(\int_{\Rt\setminus\{0\}} |f|^2r^{-2\delta-3}dx\right)^{1/2},
\end{equation}
and
\begin{equation}
\label{eq:38}
\Vert f\Vert_{H'^{2}_\delta}:=\sum_{0}^{2}\Vert D^jf\Vert_{L'^{2}_{\delta-j}}.
\end{equation}

The advantage of using these spaces is that they deal with weights both at infinity and at the origin, and therefore they include functions with certain fall-off properties at infinity and which are divergent at $r=0$. Since our main functions will be singular at the origin, we can not use standard $H^2$ or $H^2_\delta$ Sobolev spaces (see \cite{Bartnik86} for more details and properties of all these spaces). In particular, it is important to remark that if a function $f\in H^ {'2}_{\delta}$ then  $f=o(r^ \delta)$ at infinity and at the origin (see \cite{dain-gabach10} for the proof). 
 
We use these spaces in the statement of our main result:

\begin{theorem}
\label{teorema}
Let $J$, $A$ and $P$ be non negative constants with $J$ or $A$ different from zero, $K$ be given by (\ref{kgeneral}), and $\rho$ satisfy (\ref{rho}). Then, there exists a unique solution $u_0\in H'^2_\delta$, $\delta\in(-1,-1/2)$ of equation \eqref{ecgeneralu0} such that  $u_0$ is $C^\infty$ in $\Rt\setminus\{0\}$, and it can be written as
\begin{equation}\label{ref}
 u_0(x)=\frac{V(\theta,\phi)}{\sqrt{r}(1+b\sqrt{r})}+U(x)
\end{equation}
where $V\in C^{\infty}(S^2)$ is a positive function depending only on $J$ and $A$, $b$ is a positive, fixed constant, possibly depending on $A,\, J$ and $P$, and $U\in H'^2_{-1/2}$ is $o(r^{-1/2})$ near the origin.
\par
Moreover, we have that $u_0$ is the
limit of the sequence
\begin{equation}
  \label{eq:4}
  \lim_{\mu \to 0}  u_\mu=u_0,
\end{equation}
in the norm $H'^{2}_\delta$.
\end{theorem}

Before getting into the  proof, let us discuss the main issues this theorem exposes. 

This result proves the existence of an extreme, non vacuum, conformally flat initial data family, parametrized by the angular momentum $J$ of the data, its linear momentum $P$ and the constant $A$. We remark that we do not require both $A$ and $J$ to be non zero, but only assume one of them is non vanishing. The very well known extreme Reissner N\"ordstrom data is included in this family by setting $J=A=P\equiv0$ and $\sigma=q^2$ where $q$ is the electric charge on the slice. It also includes extreme spinning Bowen-York data (\cite{dain-gabach09}) by setting $A=P=\sigma\equiv0$ and non zero angular momentum, and also a particular foliation of Schwarzschild's black hole (see \cite{hhobm}) when $P=J=\sigma\equiv0$. On the other hand, This theorem does not include Kerr's initial data, since it has been proven \cite{kroon} that there exists no foliation of Kerr spacetime (including extreme Kerr) being conformally flat.

The main feature we observe by performing this limiting procedure is the change in the global structure of the initial data. The asymptotic geometry moves from having two asymptotically flat ends, to having one asymptotically flat end and one cylindrical end. This property is, of course translated into special properties in the physical 3-metric $\bar g$, which shows an asymptotic cylindrical nature at the origin and the usual fall-off as $r\to \infty$.

Note that $u_0\in H'^2_{\delta}$ implies that the limit function is a strong solution of equation (\ref{ecgeneralu0}) also at the origin. The first term in the decomposition \eqref{ref} completely determines  the asymptotic geometry of the end at $r=0$, while the fact that $U\in H'^2_{-1/2}$ implies that this function is $o(r^{-1/2})$ at $r\to 0$ and thereby does not contribute to this feature. The behavior $O(r^{-1/2})$ of the first term in \eqref{ref} near the origin is  responsible for the cylindrical nature of the end, since the physical metric has the asymptotic form
\begin{equation}
\bar g_{ij}\approx \frac{V(\theta,\phi)}{r^ 2}\delta_{ij} \qquad\mbox{as}\quad r\to 0
\end{equation}
and therefore it especifies its limiting sectional area $\mathcal A_0$ through the expression
\begin{equation}
\mathcal{A}_0:=\lim_{r\to0}\mathcal{A}_r=\lim_{r\to0}\oint_{B_r}\Phi_0^4r^2\sin\theta d\theta d\phi=\int_0^\pi\int_0^{2\pi}V^4\sin\theta d\theta d\phi,
\end{equation}
where $B_r$ is a ball of radius $r$ centered at the origin and $\mathcal A_r$ its area. Since $V$ is a strictly positive, bounded function on $S^2$, this area is finite and  different from  zero. Clearly, in the non extreme case ($\mu>0$), $\Phi_\mu=O(r^{-1})$ as $r\to0$ (see equation \eqref{fimu}), and therefore $\mathcal{A}_0\to\infty$ showing that the origin is an asymptotically flat end.  As it was mentioned in \cite{dain-gabach09}, the same phenomenon occurs as we take the extreme limit in Reissner N\"ordstrom, Kerr and spinning Bowen-York initial data.

Moreover, in the extreme case studied here, due to the equation satisfied by the function $V$ (see below, equation \eqref{ecuv}), this area is parametrized by the angular momentum, the constant $A$ and possibly matter, while the linear momentum does not play any rol in this end (note also, from equations \eqref{kj}, \eqref{ka} and \eqref{kp}, that $K_A$ and $K_J$ diverge as $r^{-3}$ near the origin, while $K_P$ does  as $r^{-2}$).

The proof of this theorem will be divided in two parts, an existence proof, in section 3 and a uniqueness proof, presented in section 4.

\section{Existence}
The plan of the existence proof is as follows: We first prove that the sequence $u_\mu$ is pointwise monotonically increasing as $\mu$ decreases (Section 3.1). Then, we show that there exists a function $u_0^+$, independent of $\mu$, which is an upper bound to this sequence for all $\mu$ (Section 3.2).  From this upper bound we construct a lower bound $u_0^-$.  Finally, we prove convergence in the appropriate functional space (Section 3.3). 

Due to the similarity of the equations involved, this existence proof follows along the lines presented in \cite{dain-gabach09} for the Bowen-York spinning case, therefore, we omit here some points and refer the reader to \cite{dain-gabach09} for more details about the arguments employed.

\subsection{Monotonicity}
We first show that if $\mu_1\geq \mu_2>0$ then $u_{\mu_1}(x)\leq u_{\mu_2}(x)$ for all $x\in\Rt$. 
 
Define  $w$ by
\begin{equation}\label{w}
w(x)=u_{\mu_2}(x)-u_{\mu_1}(x),
\end{equation}
then using equation (\ref{ecgeneralumu}) and the nondecreasing property of the function $F$ defined in \eqref{lich}, we obtain that $w$ satisfies 

\begin{equation}\label{ww}
\Delta w-w H=\frac{\mu_2-\mu_1}{2r}H,
\end{equation}
where $H=H(\Phi_2,\Phi_1)=H(\Phi_1,\Phi_2)$ is a non negative function given by
\begin{equation} 
\label{eq:H}
H(\Phi_2,\Phi_1)=\frac{K^2}{8}\sum_{i=0}^{6}
\Phi_{1}^{i-7}\Phi_{2}^{-1-i}+2\pi\rho\sum_{i=0}^{2}
\Phi_{1}^{i-3}\Phi_{2}^{-1-i}.
\end{equation}
Since $u_\mu\geq 0$ for $\mu >0$, we have the bound $H(\Phi_1,\Phi_2)\leq H(1+\frac{\mu_1}{r},1+\frac{\mu_2}{r})$ which is finite for all $x\in\Rt$ when $\mu_1,\mu_2>0$.

Since $H\geq 0$ and by hypothesis
we have $\mu_2-\mu_1\leq 0$, then the right hand side of (\ref{ww}) is
negative. We also have that $w\to 0$ as $r\to \infty$, hence, we can apply the Maximum Principle for the
Laplace operator to conclude that $w\geq0$ in
$\Rt$.  We emphasize that this theorem can be applied because $H$ is
bounded in $\Rt$ when $\mu_1,\mu_2 > 0$ .
 
Remarkably, the sequence $\Phi_\mu$ has the opposite behavior as the
sequence $u_\mu$, namely  $\Phi_\mu$ is pointwise increasing with respect to
$\mu$. This can be proven using the Maximum Principle and the fact that $u_\mu$ is bounded at the origin, whereas $\mu/2r$ is not (see Lemma 3.2 in \cite{dain-gabach09} for the complete proof in a similar situation).
  \subsection{Bounds}
  \begin{lemma}
\label{l:rn}
(Upper bound) Let $Q$ be a positive constant such that
\begin{equation}
\label{condicion}
Q^2\geq \sigma+7P^ 2+3A+9J.
\end{equation}
Then for all $\mu>0$  we have
\begin{equation}\label{desigrn}
u_{\mu}(x)\leq u^+_{\mu}(x) <  u^+_{0}(x),
\end{equation}
where  
\begin{equation}\label{urn}
 u^+_\mu=\sqrt{1+\frac{M}{r}+\frac{\mu^2}{4r^2}}-1-\frac{\mu}{2r},\quad\quad M:=\sqrt{Q^2+\mu^2}
\end{equation}
and 
\begin{equation}\label{urncero}
 u^+_0=\sqrt{1+\frac{Q}{r}}-1.
\end{equation}
\end{lemma}

\begin{proof}
In order to prove the first inequality of \eqref{desigrn}, we compute from (\ref{urn}) 
\begin{equation}
 I:=\Delta u^+_{\mu}- F(x,u_\mu^+)=
   -\frac{Q^2}{4r^4\left(  \Phi_\mu^+ \right)^3}+\frac{K^2}{8(\Phi_\mu^+)^7}+\frac{\sigma}{4r^4(\Phi_\mu^+)^3}
\end{equation}
where $\Phi_\mu^+:=1+\mu/2r+u_\mu^+$ is Reissner N\"ordstrom conformal factor for the usual foliation $t=$const, and $r$, the isotropical radius. Now,  assuming that condition (\ref{condicion}) holds, we write
\begin{equation}
I\leq\frac{1}{4r^4(\Phi_\mu^+)^3}\left(\frac{K^2r^4}{2(\Phi_\mu^+)^4}-(7P^ 2+3A+9J)\right).
\end{equation}
Next, we use $\Phi_\mu^+\geq\Phi_0^+=\sqrt{1+Q/r}$ and the following bound for $K^2$ 
\begin{equation}
K^2\leq\frac{6A^2+18J^2}{r^6}+\frac{12AP+18PJ}{r^5}+\frac{27P^2}{2r^4}
\end{equation}
to find $I\leq0$, which implies  
\begin{equation}\label{eq:u++b}
\Delta u_\mu^+\leq F(x,u_\mu^+).
\end{equation}
Now, we define the difference $w= u^+_\mu- u_\mu$, and using equation (\ref{ecgeneralumu})  and (\ref{eq:u++b})  we obtain
\begin{equation}
  \label{eq:7}
  \Delta w -w H(\Phi^+_\mu, \Phi_\mu) \leq 0.
\end{equation}
Note that the function $w$ is not $C^2$ at the origin because 
$u^+_\mu$ is not $C^2$ (and in general, $u_\mu$ is neither), and hence it does not satisfy inequality
\eqref{eq:7} in the classical sense at the origin. 
However, we have  $w\in H_{loc}^1$ (in fact $w$ is $C^1$) and then it satisfies   (\ref{eq:7}) in the weak
sense also at the origin. We also have that  $w$ goes to zero as $r\to
\infty$. Hence, we can apply the Maximum Principle to conclude that $w\geq 0$, i.e $u_\mu^+\geq u_\mu$. 

Finally, inequality $u_\mu^+< u_0^ +$ for $\mu>0$ can be checked directly from the explicit expressions (\ref{urn})-(\ref{urncero}).
\end{proof}

\begin{lemma}
\label{l:rnsub}
(Lower bound) Let $u^-_\mu$ with $\mu\geq0$ be the solution to the linear equation
\begin{equation}
\Delta u_\mu^{-}=-\frac{K^2}{8(\Phi_\mu^+)^7}.
\end{equation}
We have that for all $\mu > 0$
\begin{equation}
  \label{eq:14}
 u^-_\mu(x) \leq  u_\mu(x)
\end{equation}
and
\begin{equation}
  \label{eq:33}
  \Phi_\mu^-(x)  \geq  \Phi^-_0(x),  \quad \mbox{where}\quad \Phi_\mu^-:=1+\frac{\mu}{2r}+u_\mu^-.
\end{equation}
In addition, 
\begin{equation}\label{compcerosub}
u_0^-=O(r^{-1/2})\quad \mbox{as}\quad r\to0.
\end{equation}
\end{lemma}
\begin{proof}
The solution $u_\mu^-$ for all $\mu\geq0$ can be explicitly constructed  using the fundamental
solution of the Laplace operator. From the standard elliptic
estimates we deduce that $u^-_\mu \in C^{2,\alpha}(\Rt)$ for $\mu>0$.

Let us prove inequality (\ref{eq:14}).  As usual we take the
difference $w=u_\mu-u^-_\mu$, then we
have
\begin{equation}
  \label{eq:16}
  \Delta w = F(x, \Phi_\mu)- F(x, \Phi^+_\mu)=(u_\mu-u^+_\mu)H(\Phi_\mu,\Phi^+_\mu).
\end{equation}
Since $u_\mu-u^+_\mu\leq 0$ by lemma \ref{l:rn} and $H\geq0$, we obtain  $\Delta
w\leq 0$, with $w\to0$ at infinity, thereby, the maximum principle gives $w=u_\mu-u^-_\mu\geq 0$.

Inequality (\ref{eq:33}) can be verified using Lemma 3.2. of \cite{dain-gabach09}, or by explicit means. 

Finally, in order to check the fall off behavior of the lower bound $u_0^-$ which satisfies
\begin{equation}\label{ecsubc}
\Delta u_0^{-}=-\frac{K^2}{8(\Phi_0^+)^7}
\end{equation}
we write it as
\begin{equation}
u_0^-=\frac{V^-}{\sqrt{r}(1+b\sqrt r)}+U^-,
\end{equation}
where $b$ is a positive, fixed constant,
\begin{equation}\label{vmenossuma}
V^-=\frac{3}{Q^{7/2}}\left(A^2+\frac{J^2}{25}(51-3\sin^2\theta)\right)\geq\frac{3}{Q^{7/2}}\left(A^2+\frac{46}{25}J^2\right)>0
\end{equation}
and $U^-$ solves the remaining linear equation
\begin{equation}\label{ezc}
\Delta U^-=\Delta \left(u_0^--\frac{V^-}{\sqrt{r}(1+b\sqrt r)}\right):=\ell(x).
\end{equation}
By an explicit calculation it can be seen that $\ell\in L^{'2}_{-5/2}$, and since the Laplace operator is an isomorphism $\Delta:H^{'2}_{-1/2}\to L^{'2}_{-5/2}$ (see \cite{Bartnik86}) we find $U^-\in H^{'2}_{-1/2}$. Therefore $U^{-}$ is $o(r^{-1/2})$ near the origin, and we have proven the fall off behavior of $u_{0}^-$ in \eqref{compcerosub}. We remark, however, that $U^-$ can be found in explicit, closed form in terms of a few spherical harmonics, although it is not necessary for our purposes here.
\end{proof}

\subsection{Convergence} 

In this section we prove that the sequence $u_\mu$ is Cauchy in the norm $H'^{2}_\delta$ for $-1<\delta<-1/2$, which translates into proving 
\begin{equation}
\lim_{\mu_1,\mu_2\to 0}\Vert u_{\mu_1}-u_{\mu_2} \Vert_{H'^2_{\delta}}=0
\end{equation}

Consider the sequence $u_{\mu}r^{-\delta-3/2}$ for
$-1<\delta<-1/2$ . This sequence is pointwise bounded by
$u^+_{0}r^{-\delta-3/2}$ and monotonically increasing
as the parameter $\mu$ goes to zero, which means that it is
a.e. pointwise converging to a function $u_0r^{-\delta-3/2}$. Then, since
$u^+_{0}r^{-\delta-3/2}$ is summable in $\Rt$ for the given values of
the weight $\delta$, we find that the new sequence converges in
$L^2(\Rt)$. But this implies that the original sequence $u_{\mu}$
converges in $L'^{2}_\delta$, with $\delta\in(-1,-1/2)$. That is
\begin{equation}
\label{limwr3}
\lim_{\mu_1,\mu_2\rightarrow 0}\Vert w\Vert_{L'^{2}_\delta}=0,
\end{equation}
where  $w:=u_{\mu_1}-u_{\mu_2}$ and for convenience we set $\mu_1\geq\mu_2$.

In order to prove that the sequence $u_{\mu}$ is a Cauchy
sequence also in the weighted Sobolev space $H'^{2}_\delta$ with
$\delta\in(-1,-1/2)$, we will apply the following estimate (see \cite{Bartnik86})
\begin{equation}\label{estimacion1}
\Vert w\Vert _{H'^{2}_\delta}\leq C\Vert\Delta w\Vert_{L'^{2}_{\delta-2}},
\end{equation}
where the constant $C$ depends only on $\delta$.

From the equations satisfied by $u_{\mu_1}$ and $u_{\mu_2}$ we compute
\begin{align}
\label{delta}
  \Vert\Delta w\Vert_{L'^{2}_{\delta-2}} & =\left\Vert
    Hw+H\frac{\mu_2-\mu_1}{r}\right\Vert_{L'^{2}_{\delta-2}} \\
 & \leq
 \Vert Hw\Vert_{L'^{2}_{\delta-2}}+(\mu_1-\mu_2)\left\Vert\frac{H}{r}\right\Vert_{L'^{2}_{\delta-2}}.
\end{align}
where $H=H(\Phi_{\mu_1},\Phi_{\mu_2})$. From the definition of the norm $L'^{2}_{\delta}$ given in (\ref{eq:35})
we obtain
\begin{equation}
  \label{eq:36}
   \Vert Hw\Vert_{L'^{2}_{\delta-2}}\leq \sup_{\Rt}|Hr^2| \,  \Vert w\Vert_{L'^{2}_{\delta}},
\end{equation}
and hence, using the explicit expression of $H$ it can be seen that
\begin{equation}
 H(\Phi_{\mu_1},\Phi_{\mu_2})\leq H(\Phi_0^-,\Phi_0^-),
 \end{equation}
and thereby $Hr^2$ is bounded in $\mathbb R^3$ and the norm of $H/r$ is finite for $\delta\in(-1,-1/2)$.  Then, we can write
\begin{equation}
\Vert w\Vert_{H'^{2}_{\delta}}\leq C \left( \Vert w\Vert_{L'^{2}_{\delta}}+(\mu_1-\mu_2)\right),
\end{equation}
where the constant $C$ does not depend on $\mu$.  This and equation
(\ref{limwr3}) give 
\begin{equation}
\lim_{\mu_1,\mu_2\rightarrow 0}\Vert w\Vert_{H'^{2}_{\delta}}=0,
\end{equation}
showing that the sequence $u_{\mu}$ is Cauchy in the $H'^{2}_{\delta}$-norm, with
$\delta\in(-1,-1/2)$.

In this manner we have completed the existence proof of solution to \eqref{ecgeneralu0} in the Sobolev space $H^{'2}_\delta$.
\section{Uniqueness}
In this section we prove that the solution found above by the limit procedure, $u_0=\lim_{\mu\to0}u_\mu$, is the unique solution to equation \eqref{ecgeneralu0}.

The strategy is the following. Given a solution $u\in H^{'2}_{\delta}$ of (\ref{ecgeneralu0}), we first show that it can be uniquely decomposed as
\begin{equation}\label{dec}
 u=\frac{V}{\sqrt r(1+b\sqrt r)}+U
\end{equation}
where $V$ is a $C^\infty$ function on the 2-sphere $S^2$ (Lemma \ref{lemav}), $b$ is a positive, fixed constant, possibly depending on $A,\, J$ and $P$, and $U\in H^{'2}_{-1/2}$ is unique for a given solution $u$ (Lemma \ref{lemau}). Then, we show that if two such solutions $u$ exist, they must be equal (Lema \ref{lemauni}).

The decomposition \eqref{dec}, together with the associated equation for $V$, needed in the uniqueness proof, were inspired by the work of Hannam, Husa and O'Murchadha \cite{HHM09}, where they assume an expansion of $\Phi_0$ valid near the origin in the form $\Phi_0=D(\theta)/\sqrt r+O(\sqrt r)$ and analyze a similar equation for vacuum, axisymmetric initial data and for $A\neq0$.

Let us deal first with $V$. Define $V$ as the solution to 
\begin{equation}\label{ecuv}
\tilde{\Delta}V-\frac{1}{4}V=-\frac{3A^2+9J^2\sin^2\theta}{4V^7}-\frac{\sigma_0}{4V^3}\quad\mbox{in}\quad S^2
\end{equation}
 where $\sigma_0:=\sigma(r=0,\theta,\phi)$ is a bounded smooth function  and $\tilde\Delta$ is the Laplace-Beltrami operator on $S^2$:
\begin{equation}
 \tilde\Delta V:=\frac{1}{\sin\theta}\partial_\theta(\sin\theta\partial_\theta V)+\frac{1}{\sin^2\theta}\partial^2_\phi V.
\end{equation}

We write equation \eqref{ecuv} as $LV=g(V)$ where 
\begin{equation}
 L:=\tilde\Delta-\frac{1}{4},\qquad g(V):=-\frac{3A^2+9J^2\sin^2\theta}{4V^7}-\frac{\sigma_0}{4V^3}.
\end{equation}
and we obtain the  following lemma:
\begin{lemma}\label{lemav}
 There exists a unique positive solution $V\in C^\infty(S^2)$ to equation $LV=g(V)$, such that
\begin{equation}
 c_-\leq V\leq c_+,\qquad c_-,c_+ \quad\mbox{positive constants}
\end{equation}

\end{lemma}
\begin{proof}
We will find a sub and a supersolution, both positive $C^\infty$ functions on $S^2$ for equation $LV=g(V)$, so that the Sub-super solution theorem \cite{Isenberg95} gives existence and uniqueness of solution to this equation.
 
 Let us check that the constant $c_+$ given by
\begin{equation}\label{eccmas}
 c_+:=\left( 2\sup_{S^2}(\sigma_{0})+\sqrt{6A^2+18J^2}\right)^{1/4}
\end{equation}
is a supersolution for the operator $L$. Applying the operator $L$ to $c_+$ we have
\begin{equation}
 Lc_+=-\frac{c_+}{4}=-\frac{c_+^ 4}{8c_+^3}-\frac{c_+^ 8}{8c_+^7}\leq -\frac{\sigma_{0}}{4c_+^3}-\frac{3A^2+9J^2\sin^2\theta}{4c_+^7}=g(c_+)
\end{equation}
which shows that $c_+$ is a supersolution. 

Now, let us check that the $C^\infty$ function $v_-$ defined by
\begin{equation}\label{v}
 v_-=\frac{3}{c_+^7}\left(A^2+\frac{J^2}{25}(51-3\sin^2\theta)\right)>\frac{3}{c_+^7}(A^2+J^2):=c_->0
\end{equation}
is a subsolution for equation $LV=g(V)$. Note that $v_-=V^-Q^{7/2}/c_+^7$, and $V^-$ was used in Lemma 3.2.

It can easily be  checked that $v_-<c_+$, then we compute, using the explicit expression \eqref{v} 
\begin{equation}
Lv_-=-\frac{3A^2+9J^2\sin^2\theta}{4c_+^7}\geq g(c_+)\geq g(v_-)
\end{equation}
which shows that $v_-$ is indeed the desired subsolution.

Finally, using the Sub-Super solutions theorem \cite{Isenberg95}, we find that there exists a unique positive $C^\infty(S^ 2)$ solution $V$ to equation (\ref{ecuv}) satisfying
\begin{equation}
 0<c_-\leq v_-\leq V\leq c_+.
\end{equation}
and the lemma is proven.

Note that by the strong maximum principle \cite{Evans}, we know a priori that there exists certain positive constant $c_-$ such that $0<c_-\leq V$. However, a constant function is not a subsolution for $LV=g(V)$ unless we explicitely assume $A\neq0$.
\end{proof}

Now we treat the function $U$. We will define it as the solution to equation
\begin{equation}
\Delta U=\Delta\left(u-\frac{V}{\sqrt r(1+b\sqrt r)}\right)
\end{equation} 
where $b$ is a positive, fixed constant introduced here in order to give the right units.

It is clear that, since $u\in H^{'2}_\delta$ and $V\in C^\infty(S^2)$ exist, then $U$ exists and belongs to $ H^{'2}_{ \delta}$. We want to verify that it is unique for each $u$ and that it actually belongs to $H^{'2}_{-1/2}$, which would imply that $U$ diverges as $o(r^{-1/2})$ near the origin.

For that purpose, using the equations satisfied by $u$ and $V$ we can write the above equation in the form
\begin{equation}\label{ecU}
\left(\Delta -\frac{h}{4r^2}\right)U=f
\end{equation}
where $h$ and $f$ do not depend on $U$ and are given by
\begin{equation}\label{h}
h=h_1+h_2
\end{equation}
with
\begin{equation}
h_1:=(3A^2+9J^2\sin^2\theta)\sum_{i=0}^{6}\frac{(\sqrt r+\sqrt ru)^{i-7}}{[(1+b\sqrt r)^ {1/7}V]^{i+1}}
\end{equation}
\begin{equation}\label{h1}
h_2:=\sigma_0\sum_{i=0}^{2}\frac{(\sqrt r+\sqrt ru)^{i-3}}{[(1+b\sqrt r)^{1/3}V]^{i+1}}
\end{equation}
and 
\begin{eqnarray}\label{deff}
 f:=\frac{6AP^an_a+9\epsilon_{abc}n^aP^bJ^c}{4r^{5}(1+u)^7}-\frac{\sigma-\sigma_0}{4r^{4}(1+u)^3}-\frac{9(P^2+2(P^an_a)^2)}{16r^{4}(1+u)^7}\nonumber\\
 +\frac{h_1}{4r^{5/2}}\left[\sqrt r+V\left(\frac{1}{1+b\sqrt r}-(1+b\sqrt r)^{1/7}\right)\right]+\nonumber\\\frac{h_2 V}{4r^{5/2}}\left(\frac{1}{1+b\sqrt r}-(1+b\sqrt r)^{1/3}\right)+\frac{Vb}{4r^2(1+b\sqrt r)^3}(1-b\sqrt r)
\end{eqnarray}
   
Note that,  due to the known behavior  $u=O(r^{-1/2})$ at the origin, and the positivity of $V$,  $h$ is bounded in $\Rt$ and $f\in L^{'2}_{-5/2}$ (we use $\sigma-\sigma_0\to0$ as $r\to0$).

 Then we prove the following lemma

\begin{lemma}\label{lemau}
Let $u\in H^{'2}_{\delta}$ be an extreme solution to \eqref{ecgeneralu0}. Define $V$ by Lemma 4.1. and $U$ as the solution to \eqref{ecU} with $h$ and $f\in L^{'2}_{-5/2}$ given in (\ref{h})-(\ref{deff}). 
Then we have that for each $u$, the solution $U$ is unique and $U\in H^{'2}_{-1/2}$.
\end{lemma}
\begin{proof}
We write equation \eqref{ecU} as $\mathcal L U=f$ with $\mathcal L:=\Delta -\frac{h}{4r^2}$. Then since the operator $\mathcal L$ is an isomorphism $H'^2_{-1/2}\to L^{2}_{-5/2}$ (see the Appendix for the proof), we obtain that for each $f$ (that is, for each solution $u$ to \eqref{ecgeneralu0}) there exists a unique $U\in H^{'2}_{-1/2}$ satisfying the above equation. 
\end{proof}

With these two lemmas, we have
\begin{equation}
\Delta\left[u-\left(\frac{V}{\sqrt r(1+b\sqrt r)}+U\right)\right]=0
\end{equation}
and since the Laplace operator is an isomorphism $H^{'2}_{\delta}\to L^{'2}_{\delta-2}$ we conclude
\begin{equation}
u=\frac{V}{\sqrt r(1+b\sqrt r)}+U.
\end{equation}

With these results, we are now ready  to prove uniqueness of the extreme solution. 

\begin{lemma}\label{lemauni}
 The solution $u$ to equation \eqref{ecgeneralu0} is unique in $H^{'2}_{\delta}$.
\end{lemma}
\begin{proof}
Assume, on the contrary, that there exist two such solutions, $u$ and $\tilde u$, and write
\begin{equation}
u=\frac{V}{\sqrt{r}(1+b\sqrt r)}+U,\quad \tilde u=\frac{V}{\sqrt{r}(1+b\sqrt r)}+\tilde U
\end{equation}
(note that the same $V$ appears in both solutions) and define the difference 
\begin{equation}
w=u-\tilde u=U-\tilde U.
\end{equation}
We see that $w\in H^{'2}_{-1/2}$, because $U$ and $\tilde U$ do, by Lemma 4.2. Then, in virtue of the equations satisfied by $u,\,\tilde u$ we have
\begin{eqnarray}
\Delta w&=&-\frac{K^2}{8}\left(\frac{1}{(1+u)^7}-\frac{1}{(1+\tilde u)^7}\right)-\frac{\sigma}{4r^4}\left(\frac{1}{(1+u)^3}-\frac{1}{(1+\tilde u)^3}\right)\nonumber\\
&=&H\left(1+u\,,\, 1+\tilde u \right) w
\end{eqnarray}
where $H\geq0$ was defined in \eqref{eq:H}. But from the explicit expression, we have that $H=O(r^{-2})$ at $r\to0$, and goes to zero at infinity, therefore, we can apply the theorem from the appendix ($\Delta-H(1+u,1+\tilde u):H^{'2}_{-1/2}\to L^{'2}_{-5/2}$ is an isomorphism) to conclude that $w\equiv0$, and thereby $u\equiv \tilde u$.
\end{proof}
This result completes the proof of our main result, Theorem 2.1. on existence and uniqueness of the extreme solution to \eqref{ecgeneralu0}
\section{Final Comments}

In this work we have learnt that given a conformally flat, maximal family of non extreme black hole initial data for Einstein equation (parametrized by the parameter $\mu>0$), having angular and linear momentum and possibly some types of matter, there always exists a special and singular limit ($\mu=0$), called the extreme initial data which has a completely different geometry than the original family. Namely, while each non extreme data in the family has a wormhole-like geometry (two asymptotically flat ends), the extreme limit has one asymptotically flat end and one cylindrical end. Moreover, this change may be produced by the angular momentum, matter or the presence of other singular term in the constraint equation (the term containing the constant $A$). Any one of these factors alone  can transform one asymptotically flat end into a cylindrical end. On the other hand, the linear momentum of the data plays no rol in making this transition, since it can not produce the desired behavior at the origin. 

We remark here that the observed behavior near $r=0$ of the solution $u_0$ is not present when we let $G$ be non zero. This would amount to saying that the end at the origin has non-zero linear momentum. It also fails to be true when we deal with the vacuum case and vanishing $J$ and $A$, this is, when the data only possesses linear momentum, since there is no term producing the cylindrical infinity. 

As we mentioned in section 2, Theorem \ref{teorema} can also be applied with no mayor modifications to tensors $K$ including a term $K_\lambda$ for appropriate complex functions $\lambda$. Nevetheless, calculations become much more involved and do not seem to bring out any new insight on the underlying phenomena. 

As opposed to what happens with the end at $r\to0$, the asymptotic geometry of the other end, at $r\to\infty$, does not seem to suffer any relevant change. It remains being an asymptotically flat end. In this respect, it would be interesting to analyze what is the effect of taking the extreme limit on the ADM mass. We know (see \cite{dain-gabach09} for details) that if the electric charge $q$ is held fixed in a Reissner N\"ordstrom initial data, then the total mass $ m$ decreases as $\mu$ goes to zero, reaching its minimum value at the extreme data,\textit{ i.e.}, when $m=q$. The same phenomenon is seen in Kerr's data when the total angular momentum $J$ is fixed. And finally, it has also been numerically indicated in the spinning Bowen-York initial data, that the total mass is a minimum in the extreme case $\mu=0$. For the present initial data family, in order to explore this issue further one should have information, at least, about the behavior of the radial derivative $\partial_r u_\mu$ near infinity. Nevertheless, we expect that in this general case too, the total mass decreases as we approach the extreme limit $\mu=0$, which would correspond to reaching the initial data with maximum amount of angular momentum and matter per unit mass. If this were the case, then the name "extreme" would be fully consistent with the familiar notions we take from extreme Kerr and extreme Reissner N\"ordstrom black holes.

The case of non conformally flat initial data seems to be more difficult if the present limiting procedure is attempted. First, because of the presence of the Ricci scalar in Lichnerowicz equation, which in general does not have a definite sign, and even might depend on the parameter $\mu$, as in the case of Kerr initial data. This could complicate the task of finding appropriate bounds for the non extreme solutions $u_\mu$. And second, because it is not easy to find, in the literature, basic mathematical results as the Maximum principle, or the statement on the non flat Laplace operator being an isomorphism between weighted Sobolev spaces. This is due to singular behavior of the functions and equations involved. However we believe that the case of axial symmetry could be approached in this way, and that it could be a useful, though laborious tool in the study of pertubations of extreme Kerr initial data. See \cite{dain-gabach10} for a different approach to the problem of small deformations of extreme Kerr black hole initial data.

As a final comment, we want to remark that there are two situations in which some steps in the proof of Theorem 2.1 become somewhat easier. One is when $\sigma_0$ (i.e. the value of the matter function $\sigma$ at $r=0$) is a strictly positive function and the other is when $A\neq0$. In both cases we can construct appropriate lower bounds for $u_\mu$ (and also for $V$ in section 4) much more easily. For instance when $\sigma\geq a>0$, where $a$ is some constant, the lower bound $u_\mu^-$ can be taken just as 
\begin{equation}
u^-_\mu=\sqrt{1+\frac{m}{r}+\frac{\mu^2}{4r^2}}-1-\frac{\mu}{2r}, \quad m:=\sqrt{a+\mu^2}.
\end{equation}  
This indeed is what occurs, i.e., when matter consist of an electromagnetic field associated to an point electric charge $q$ (in this case $\sigma=q^2$). When $A\neq 0$, the subsolution $\Phi_\mu^-$ can be taken as the conformal factor corresponding to Schwarzschild black hole's initial data (see \cite{HHM09} for details).
Moreover, due to these observations, when $A$ or $\sigma$ are not zero, the term containing the angular momentum becomes irrelevant in the calculations.

\section*{Acknowledgments}

It is pleasure to thank my advisor, Sergio Dain, for many useful discussions, suggestions and ideas concerning this article.
I also want to thank Robert Beig for interesting discussions that took place at FaMAF during his visit in November 2009, and Niall O'Murchadha, Sascha Husa and Mark Hannam for discussions.

This work was supported in part by grants PIP 6354/05 and PIP 112-200801-00754 of CONICET (Argentina) and the Partner Group grant of the
Max Planck Institute for Gravitational Physics, Albert Einstein
Institute (Germany). 

\section{$\mathcal L$ is an isomorphism}

 In this section we prove that the linear map $\mathcal L:H^{'2}_{-1/2}\to L^{'2}_{-5/2}$ is an isomorphism. This important result is used in the proof of uniqueness of the extreme solution. Also, it  turns out to be useful in a the study of perturbations of extreme Kerr initial data \cite{dain-gabach10}. 
 
 \begin{theorem}\label{teoiso}
The linear map $\mathcal L$ defined by
\begin{equation}\label{ecL}
\mathcal Lu:=-\Delta u+\alpha u=f\qquad \mbox{in}\; \mathbb R^{3}\setminus\{0\},
\end{equation}
where $h$, defined in (\ref{h}) is a bounded function on $\Rt$, and $\alpha\geq0$ is given by
\begin{eqnarray}\label{alpha}
\alpha:=\frac{h}{4r^2},
\end{eqnarray}
is an isomorphism $H^{'2}_{-1/2}\to L^{'2}_{-5/2}$.
 \end{theorem}
 
 We decompose the proof into two parts. First, we prove the existence of a weak solution ( Lemma \ref{lemadebil}), and then, we find it to be regular (in Lemma \ref{lemareg}).

\begin{lemma} \label{lemadebil}
There exists a unique weak solution $u\in H^{'1}_{-1/2}$ of \eqref{ecL} for each $f\in L^{'2}_{-5/2}$, where $\alpha\geq0$ is given in (\ref{alpha}).
\end{lemma}
\begin{proof}
For $u,v\in H^{'1}_{-1/2}$, define the bilinear form
\begin{equation}
B[u,v]:=\int_{\Rt}Du\cdot Dv+\alpha uvdx
\end{equation}
which corresponds to the linear operator $\mathcal Lu:=-\Delta u+\alpha u$, where $\alpha$ was defined in (\ref{alpha}).

Let us check that $B[\;,\,]$ satisfies the hypothesis of Lax-Milgram's Theorem (see \cite{Evans}).
First, we have
\begin{equation}
|B[u,v]|=\left|\int_{\Rt}Du\cdot Dv+\alpha uv dx\right|\leq\left|\int_{\Rt}Du\cdot Dv dx\right|+\left|\int_{\Rt}\alpha uv dx\right|.
\end{equation}
Using the expression for $\alpha$ and H\"older's inequality, we obtain
\begin{eqnarray}
|B[u,v]|&\leq&|Du|_{L^2}|Dv|_{L^2}+C\left|\int_{\Rt}\frac{uv}{r^2} dx\right|\leq\nonumber\\
&\leq&|Du|_{L^2}|Dv|_{L^2}+C\left|\frac{u}{r}\right|_{L^2}\left|\frac{v}{r}\right|_{L^2}=\nonumber\\
&=&|Du|_{L^{'2}_{-3/2}}|Dv|_{L^{'2}_{-3/2}}+C\left|u\right|_{L^{'2}_{-1/2}}\left|v\right|_{L^{'2}_{-1/2}}\leq\nonumber\\
&\leq&\max\{1,C\}\left(|Du|_{L^{'2}_{-3/2}}|Dv|_{L^{'2}_{-3/2}}+\left|u\right|_{L^{'2}_{-1/2}}\left|v\right|_{L^{'2}_{-1/2}}\right)\leq\nonumber\\
&=&\max\{1,C\}|u|_{H^{'1}_{-1/2}}|v|_{H^{'1}_{-1/2}}
\end{eqnarray}
thereby verifying the suryectivity hypothesis of Lax-Milgram's Theorem.

Let us move now to the coercitivity condition. We have
\begin{equation}
B[u,u]=\int_{\Rt}(Du)^2+\alpha u^2dx=\int_{\Rt}(Du)^2dx+\int_{\Rt}\alpha u^2dx.
\end{equation}
Now, using the fact that $\alpha$ is non-negative, and Poincare's inequality (see \cite{Bartnik86}, Theorem 1.3), we get
\begin{eqnarray}
B[u,u]&\geq&\int_{\Rt}(Du)^2dx=\frac{1}{2}\int_{\Rt}(Du)^2dx+\frac{1}{2}\int_{\Rt}(Du)^2dx=\nonumber\\
&=&\frac{1}{2}|Du|_{L^2}+\frac{1}{2}|Du|_{L^2}\geq \frac{1}{2}|Du|_{L^2}+\frac{1}{4}\left|u\right|_{L^{'2}_{-1/2}}\geq\nonumber\\
&\geq&\frac{1}{4}\left(|Du|_{L^2}+\left|u\right|_{L^{'2}_{-1/2}}\right)=\frac{1}{4}|u|_{H^{'1}_{-1/2}}
\end{eqnarray}
verifying that $\mathcal L$ is a 1-1 operator.

Now, fix $f\in L^{'2}_{-5/2}$ and define $\langle f,v\rangle:=(f,v)_{L^2}$, where $(f,v)_{L^2}$ denotes the $L^2$-inner product. Let us check that this is a bounded linear functional on $L^{'2}_{-1/2}$, and therefore, on $H^{'2}_{-1/2}$. Let $v\in H^{'2}_{-1/2}$, then we have
\begin{eqnarray}
\langle f,v\rangle&=&=\int_{\Rt}fvdx=\int_{\Rt}fr^{\frac{-2(-5/2)-3}{2}}vr^{\frac{-2(-1/2)-3}{2}}dx\leq\nonumber\\
&\leq&\left|fr^{\frac{-2(-5/2)-3}{2}}\right|_{L^2}\left|vr^{\frac{-2(-1/2)-3}{2}}\right|_{L^2}=|f|_{L^{'2}_{-5/2}}|v|_{L^{'2}_{-1/2}}\leq\nonumber\\
&\leq&|f|_{L^{'2}_{-5/2}}|v|_{H^{'1}_{-1/2}}
\end{eqnarray}
as we wanted to prove.

Then with these three conditions fulfilled, Lax-Milgram's Theorem states that there exists a unique $u\in H^{'2}_{-1/2}$  such that
\begin{equation}
B[u,v]=\langle f,v\rangle,\qquad \forall v\in H^{'1}_{-1/2}, 
\end{equation}
that is, such that
\begin{equation}
\int_{\Rt}(\mathcal Lu-f)vdx=0,\qquad \forall v\in H^{'1}_{-1/2}.
\end{equation}
Therefore $u$ is the unique weak solution of $\mathcal Lu=f$. 
\end{proof}

Next, we use standard regularity theorems (see e.g \cite{Gilbarg}, chapter 8, for more details), to find that $u$ is a $C^\infty$ function in $\Rtc$. 

We will use this regularity, to prove the following lemma.
\begin{lemma}\label{lemareg}
Let $f\in L^ {'2}_{-5/2}$. Assume $u\in H^ {'1}_{-1/2}$ is a weak solution of $\mathcal Lu=f$. Then $u\in H^ {'2}_{-1/2}$
\end{lemma}
\begin{proof}
 Let $u\in H^ {'1}_{-1/2}$ be the unique weak solution to
 \begin{equation}
\mathcal L u=\Delta u-\alpha u=-f,
\end{equation}
then we verify that $\tilde f:=f-\alpha u \in L^{'2}_{-5/2}$: 
\begin{eqnarray}
|\tilde f|_{L^{'2}_{-5/2}}&=&|\alpha u-f|_{L^{'2}_{-5/2}}\leq C\left(|\alpha u|_{L^{'2}_{-5/2}}+|f|_{L^{'2}_{-5/2}}\right)\nonumber\\
&\leq&C\left(\left|\frac{u}{r^2}\right|_{L^{'2}_{-5/2}}+|f|_{L^{'2}_{-5/2}}\right)=C\left(\left|u\right|_{L^{'2}_{-1/2}}+|f|_{L^{'2}_{-5/2}}\right)\nonumber \\
&\leq&C\left(\left|u\right|_{H^{'1}_{-1/2}}+|f|_{L^{'2}_{-5/2}}\right).
\end{eqnarray}
Since the Laplace operator is an isomorphism $\Delta:H^{'2}_{-1/2}\to L^{'2}_{-5/2}$, \cite{Bartnik86}, then there exists a unique $\tilde u\in H^{'2}_{-1/2}$ such that
\begin{equation}
\Delta \tilde u=-f+\alpha u.
\end{equation}
But this implies that $\tilde u$ is also a weak solution to the above equation. Since, by Lemma \ref{lemadebil},  the weak solution is unique, we find that $\tilde u=u\in H^{'2}_{-1/2}$
\end{proof}

These two lemmas show that there exists a unique function $u\in H^{'2}_{-1/2}$ which solves equation $-\Delta u+\alpha u=f$ a.e, for each $f\in L^{'2}_{-5/2}$. This, in turn, means that $\mathcal L:=-\Delta+\alpha$ is an isomorphism $H^{'2}_{-1/2}\to L^{'2}_{-5/2}$, proving Theorem \ref{teoiso}.

\end{document}